\documentclass[11pt]{article}

\usepackage[T1]{fontenc}
\usepackage[sc]{mathpazo}

\usepackage{mathtools} % replaces amsmath
\usepackage{amsthm}
\usepackage{amssymb}
\usepackage{thmtools, thm-restate}

\usepackage[margin=1in]{geometry}
\usepackage{setspace}
\usepackage{relsize}
\usepackage{caption}
\usepackage{tcolorbox}
\usepackage{textcomp}

\usepackage{appendix}
\usepackage[textsize=tiny]{todonotes}

\usepackage[
    backend=biber,
    style=alphabetic,
    url=false,
    doi=false,
    isbn=false,
    sorting=nyt,
    sortcites=true,
    backref=true,
    maxbibnames=10,
    maxcitenames=3,
    mincitenames=1,
]{biblatex}
\usepackage[
    labelfont=sf,
    font=it,
    size=small,
    hypcap=false,
    format=hang,
    margin=1cm,
    justification=justified,
    calcwidth=0.8\linewidth
]{caption}

\usepackage{xcolor}
\usepackage{tikz}
\usetikzlibrary{matrix, fit, positioning, arrows.meta, decorations.pathreplacing, calc}

\usepackage{bookmark}
\usepackage{hyperref}
\hypersetup{
    hidelinks,
    colorlinks,
    linkcolor=[rgb]{0.3,0.3,0.6},
    citecolor=[rgb]{0.2,0.6,0.2},
    urlcolor=[rgb]{0.6,0.2,0.2}
}

\DefineBibliographyStrings{english}{%
  backrefpage  = {cited on page},
  backrefpages = {cited on pages}
}

\newcommand{\doiorurl}{%
  \iffieldundef{doi}
    {\iffieldundef{url}
       {}
       {\strfield{url}}}
    {http://dx.doi.org/\strfield{doi}}%
}

\newcommand{\myhref}[1]{%
  \ifboolexpr{%
    test {\ifhyperref}
    and
    not test {\iftoggle{bbx:url}}
    and
    not test {\iftoggle{bbx:doi}}
  }
    {\href{\doiorurl}{#1}}
    {#1}%
}

\DeclareFieldFormat{title}{\myhref{\mkbibemph{#1}}}
\DeclareFieldFormat
  [article,inbook,incollection,inproceedings,patent,thesis,unpublished]
  {title}{\myhref{\mkbibquote{#1\isdot}}}

\newif\ifcomments
\commentsfalse  % <-- comment this to enable comments
\ifcomments
  \renewcommand{\comment}[2]{\marginpar{\tiny{\textbf{#1: }\textit{#2}}}}
  \newcommand{\important}[1]{{\color{red}#1}}
  \newcommand{\note}[1]{{\color{blue}#1}}
\else
  \renewcommand{\comment}[2]{}
  \newcommand{\important}[1]{}
  \newcommand{\note}[1]{}
\fi

\numberwithin{equation}{section}
\declaretheoremstyle[bodyfont=\normalfont \itshape \upshape,qed=\qedsymbol]{noproofstyle}
\declaretheoremstyle[bodyfont=\normalfont \itshape]{thmstyle} 

\declaretheorem[numberlike=equation, style=thmstyle]{theorem}

\declaretheorem[name=Theorem,style=thmstyle,numbered=no]{theorem*}
\declaretheorem[numberlike=equation,style=thmstyle]{lemma}
\declaretheorem[name=Lemma,numbered=no,style=thmstyle]{lemma*}

\declaretheorem[numberlike=equation,style=thmstyle]{corollary}
\declaretheorem[name=Corollary,numbered=no,style=thmstyle]{corollary*}

\declaretheorem[numberlike=equation,style=thmstyle]{proposition}
\declaretheorem[name=Proposition,numbered=no,style=thmstyle]{proposition*}

\declaretheorem[name=Claim,numbered=no,style=thmstyle]{claim*}

\declaretheorem[name=Conjecture,numbered=no,style=thmstyle]{conjecture*}

\declaretheorem[numberlike=equation,style=thmstyle]{definition}
\declaretheorem[unnumbered,name=Definition,style=thmstyle]{definition*}

\declaretheorem[numberlike=equation,style=thmstyle]{observation}
\declaretheorem[name=Observation,numbered=no,style=thmstyle]{observation*}

\declaretheorem[numberlike=equation,style=thmstyle]{fact}
\declaretheorem[name=Fact,numbered=no,style=thmstyle]{fact*}

\declaretheoremstyle[]{defstyle}

\declaretheorem[numberlike=equation,style=remark]{remark}
\declaretheorem[unnumbered,name=Remark,style=remark]{remark*}

  {\vspace{1ex}\noindent{\em Proof.}\hspace{0.5em}\def\myproof@name{#1}}%
  {\hfill{\tiny \qed\ (\myproof@name)}\vspace{1ex}}
\newenvironment{proof-sketch}{\medskip\noindent{\em Proof Sketch.}}{\qed\bigskip}
\newenvironment{proof-attempt}{\medskip\noindent{\em Proof attempt.}}{\bigskip}

\newcommand{\inparen }[1]{\left(#1\right)}             %\inparen{x+y}  is (x+y)
\newcommand{\inbrace }[1]{\left\{#1\right\}}           %\inbrace{x+y}  is {x+y}
\newcommand{\setdef}[2]{\inbrace{{#1}\ \mid \ {#2}}}      % E.g: \setdef{x}{f(x) = 0}

\newcommand{\F}{\mathbb{F}}

\newcommand{\N}{\mathbb{N}}

\newcommand{\Z}{\mathbb{Z}}

\newcommand{\C}{\mathbb{C}}

\newcommand{\poly}{\operatorname{poly}}
\newcommand{\rank}{\operatorname{rank}}
\newcommand{\evalDim}{\operatorname{evalDim}}

\newcommand{\sym}{\mathsf{sym}}
\newcommand{\Det}{\operatorname{\mathsf{Det}}}

\newcommand{\ESym}{\operatorname{\mathsf{ESym}}}

\newcommand{\coef}{\operatorname{coef}}

\newcommand{\veca}{\mathrm{\boldsymbol{a}}}

\newcommand{\vece}{\mathrm{\boldsymbol{e}}}

\newcommand{\vecu}{\mathrm{\boldsymbol{u}}}
\newcommand{\vecv}{\mathrm{\boldsymbol{v}}}
\newcommand{\vecw}{\mathrm{\boldsymbol{w}}}
\newcommand{\vecx}{\mathrm{\boldsymbol{x}}}
\newcommand{\vecy}{\mathrm{\boldsymbol{y}}}
\newcommand{\vecz}{\mathrm{\boldsymbol{z}}}

\newcommand{\calM}{\mathcal{M}}

\newcommand{\ABP}{\mathsf{ABP}}

\newcommand{\RO}{\mathsf{roABP}}

\renewcommand{\epsilon}{\varepsilon}
\renewcommand{\tilde}{\widetilde}

\renewcommand{\epsilon}{\varepsilon}
\newcommand{\ignore}[1]{}

\title{On Closure Properties of Read-Once Oblivious Algebraic Branching Programs}

\author{
  Jules Armand\footnotemark[3] \\ \texttt{jules.armand@univ-smb.fr}
  \and
  Prateek Dwivedi\footnotemark[1] \\ \texttt{prdw@itu.dk}
  \and
  Magnus Rahbek Dalgaard Hansen\footnotemark[1] \\ \texttt{ramh@itu.dk}
  \and
  Nutan Limaye\thanks{IT University of Copenhagen, Denmark. This work was funded by the Independent Research Fund Denmark (grant agreement No.\ 10.46540/3103-00116B) and is also supported by the Basic Algorithms Research Copenhagen (BARC), funded by VILLUM Foundation Grant 54451.} \\ \texttt{nuli@itu.dk}
  \and
  Srikanth Srinivasan\thanks{Department of Computer Science, University of Copenhagen, Denmark. This work was funded by the European Research Council (ERC) under grant agreement no.\ 101125652 (ALBA).} \\ \texttt{srsr@di.ku.dk}
  \and
  S\'{e}bastien Tavenas\thanks{Universit\'{e} Savoie Mont Blanc, CNRS, LAMA. This work was supported by ANR project VONBICA (ANR-22-CE48-0007).} \\ \texttt{sebastien.tavenas@univ-smb.fr}
}

\date{}

\begin{document}
\maketitle

\begin{abstract}
\noindent We investigate the closure properties of read-once oblivious Algebraic Branching Programs (roABPs) under various natural algebraic operations and prove the following.
\begin{itemize}
    \item \textbf{Non-closure under factoring:} There is a sequence of explicit polynomials $(f_n(x_1,\ldots, x_n))_n$ that have  $\poly(n)$-sized roABPs such that some irreducible factor of $f_n$ does not have roABPs of superpolynomial size in \emph{any} order. 
    \item\textbf{Non-closure under powering:} There is a sequence of polynomials $(f_n(x_1,\ldots, x_n))_n$ with $\poly(n)$-sized roABPs such that any super-constant power of $f_n$ does not have roABPs of polynomial size in any order (and $f_n^n$ requires exponential size in any order).
    \item \textbf{Non-closure under symmetric compositions:} There are symmetric polynomials $(f_n(e_1,\ldots, e_n))_n$ that have roABPs of polynomial size such that $f_n(x_1,\ldots, x_n)$ do not have roABPs of subexponential size. (Here, $e_1,\ldots, e_n$ denote the elementary symmetric polynomials in $n$ variables.) 
    %We also show the opposite direction, namely, that  there are polynomials $(f_n(x_1,\ldots, x_n))_n$ that have roABPs of polynomial size such that their symmetric versions $f_n(e_1,\ldots, e_n)$ do not have roABPs of subexponential size. \nutan{While this sentence appears in the results, maybe it is okay to leave it out of the abstract? It is making the abstract bit too long and also this result may overshadow the other results.}
\end{itemize}
These results should be viewed in light of known results on models such as algebraic circuits, (general) algebraic branching programs, formulas and constant-depth circuits, all of which are known to be closed under these operations. 

To prove non-closure under factoring, we construct hard polynomials based on expander graphs using gadgets that lift their hardness from sparse polynomials to roABPs. For symmetric compositions, we show that the \emph{circulant} polynomial requires roABPs of exponential size in every variable order.
\end{abstract}

\section{Introduction}

Given any computational model, it is natural to study the closure properties of the model with respect to simple operations. In Boolean complexity, these simple operations typically take the form of Boolean operations such as union, intersection, complement etc. In the setting of \emph{algebraic complexity}, the object of computation is a multivariate polynomial $f\in \mathbb{F}[x_1,\ldots, x_n]$. Here, it is intuitive to consider closure properties under algebraic operations.

In this paper, we study the closure properties of a very well studied model of algebraic computation, namely \emph{read-once oblivious Algebraic Branching programs} (roABPs). The interest in this model stems from the fact that it is both expressive enough to capture many natural algorithmic paradigms while at the same time possible to analyze using standard `complexity measures'~\cite{Nisan1991}. 

In particular, $\RO$s can efficiently compute several polynomials of interest, including elementary symmetric polynomials and iterated matrix multiplication, the latter being provably hard to compute for constant-depth circuits~\cite{LST2021, BDS2022, NW1994}. In addition, $\RO$s subsume well-studied models such as sparse polynomials, set-multilinear and diagonal depth-3 circuits~\cite{KNS2020}, as well as polynomials with large partial derivative dimension~\cite{BT2024}. On the other hand, this is also one of the few  models where we have a perfect characterization of the complexity of any given polynomial (in the form of the rank of an associated matrix) and where we also have a perfect understanding of border complexity \cite{Forbes16}. As a consequence, this model has played a central role in research on lower bounds, polynomial identity testing algorithms and `debordering' results~\cite{DDSdeborder21,DS2022, BIMPS2020}.

%Despite the restriction $\RO$s can efficiently compute several polynomials of interest, including elementary symmetric polynomials and iterated matrix multiplication, which are provably hard for constant depth circuits~\cite{LST2021, BDS2022, SW2001}. In addition, $\RO$s subsume well-studied models such as sparse polynomials, set-multilinear and diagonal depth-3 circuits~\cite{KNS2020}, as well as polynomials with large partial derivative dimension~\cite{BT2024}.

We study the closure properties of this model under basic algebraic operations such as factorization, powering, and inversion under composition with an important algebraic map (the elementary symmetric polynomial map). Apart from being natural questions about any computational model, such investigations have played a vital role in understanding hardness-randomness tradeoffs~\cite{KI2004,DSY2009,CKS2019,BKRRSS2025} and the complexity of basic algebraic problems such as the Resultant and GCD~\cite{AW2024,BKRRSSgcd2025} in other algebraic models.

\subsection{Main Results}
\label{sec:main-results}

In contrast with what is known for other models, our results are mostly negative. Specifically, we show the following.

\paragraph{$\RO$ factor non-closure.} Our first main result shows that there are explicit polynomial sequences that have small roABPs but with an irreducible factor that has $\RO$ complexity super-polynomial in $n$. 
Specifically, we prove $\RO$ complexity lower bound for a \emph{root}, which is an irreducible factor of the form $x_n-f(x_1,\ldots, x_{n-1})$, even when the $\RO$ is allowed to scan the variables in any order.
The formal statement is as follows.

\begin{restatable}[$\RO$ factor non-closure]{mainthm}{roabpnonclosure}
\label{thm:roabp-non-closure}
    The following holds over any field.
    Let $n \in \N$ be a parameter and $d \geq n$.
    There exists an $n$-variate polynomial $f$ of degree $d$ computable by an $\RO$ of width $w \coloneqq 2^{O(n)}$, such that one of its (irreducible) factors $g$ requires an $\RO$ of width $w^{\Omega(\log d)}$ in every variable order.
\end{restatable}

Note that, an $\RO$ computing an $n$-variate polynomial by definition has only $n$ layers. Hence, the size and the width of an $\RO$ are polynomially related. Secondly, the size and width parameters in the theorem above are not polynomial in the number of variables, but they can be easily made polynomial by padding with some additional `dummy' variables. In particular, one should think of $n$ above as logarithmic in the number of `actual' variables and $d$ as a growing parameter, up to a polynomial in the number of variables.

This is in contrast to other algebraic models such as algebraic circuits~\cite{Kal1989, KT1990}, branching programs~\cite{ST2021}, formulas and constant-depth circuits~\cite{BKRRSS2025}, all of which satisfy the property that factors of a polynomial $f$ have complexity comparable to that of $f$. An exception to this rule is the family of sparse polynomials~\cite{GK1985}, and our construction is based on `lifting' this example to the setting of $\RO$.

\paragraph{$\RO$ complexity of Symmetric Composition.} We study an analogue of the result of Bl\"{a}ser and Jindal~\cite{BJ2019} for $\RO$. More specifically, a classical result in the theory of symmetric functions says that any symmetric polynomial\footnote{A polynomial is symmetric if it is invariant under any permutation of its variables.} $f_{\mathrm{sym}}(x_1,\ldots, x_n)$ can be written as a unique polynomial combination $f$ of the elementary symmetric polynomials $\ESym_n^1,\ldots, \ESym_n^n$, where $\ESym_n^d$ is the $n$-variate elementary symmetric polynomial of degree $d$. Looking for a computational analogue of this theorem, Lipton and Regan~\cite{LR2009}, asked: what is the complexity of $f_{\mathrm{sym}}$ vis-à-vis that of $f$?

Bl\"{a}ser and Jindal~\cite{BJ2019} showed that the complexity of $f$ and $f_{\mathrm{sym}}$ are polynomially related in the algebraic circuit model. Recently, the work of Bhattacharjee, Kumar, Rai, Ramanathan, Saptharishi and Saraf~\cite{BKRRSSgcd2025} extended this result to formulas and constant-depth circuits to show that fundamental computations such as GCD, resultants and discriminants have efficient constant-depth circuits in any characteristic. This generalizes a similar result of Andrews and Wigderson~\cite{AW2024} in characteristic~$0$.

We show in this paper that the $\RO$ complexity of a polynomial $f$ and its symmetric counterpart $f_{\sym}$ can differ significantly. Taking $f_{\mathrm{sym}} = \sum_{d = 0}^{n} \ESym_n^d(x_1^k,\ldots, x_n^k)$, we can show that $f_{\mathrm{sym}}$ is easy but $f$ is exponentially hard.

\begin{restatable}{mainthm}{symforward}
\label{thm:sym-forward}
    The following holds over fields of characteristic zero.
    Let $n \in \N$ be a parameter.
    There exists an $n$-variate polynomial $f$ such that the symmetric polynomial $f_{\sym} \coloneqq f(\ESym_n^1, \ldots, \ESym_n^n)$ is computable by an $\RO$ of constant width in every variable order, but any $\RO$ computing $f$ in any variable order must have width $2^{\Omega(n)}$.
\end{restatable}

In the other direction, our next result shows that even if a polynomial $f$ is easy to compute by $\RO$, its symmetric counterpart $f_{\sym}$ can still be hard for $\RO$—--once again in sharp contrast to the known results for circuits, formulas, and constant-depth circuits. Specifically, the lower bound for a power of the elementary symmetric polynomial yields an example where $f$ is easy but $f_{\mathrm{sym}}$ is exponentially hard.\footnote{This is an especially strong contrast to the other models where it is trivial to show that if $f$ is easy, then so is $f_{\mathrm{sym}}$.} 

\begin{restatable}{mainthm}{symbackward}
\label{thm:sym-backward}
    The following holds over fields of characteristic zero.
    Let $n \in \N$ be a parameter.
    There exists an $n$-variate polynomial $f$ computable by an $\RO$ of constant width such that its respective symmetric polynomial $f_{\sym} = f(\ESym_n^1, \ldots, \ESym_n^n)$ requires an $\RO$ of width $2^{\Omega(n)}$ in every variable order.
\end{restatable}

\paragraph{$\RO$ non-closure corollaries.} We also investigate the power of roABPs in relation to powering an efficiently computable polynomial. It is well-known that constant powers of such polynomials also have small roABPs (see e.g.~\cite[Lemma 2.5]{AFSSV2018}). However, we show that for larger powers, a superpolynomial blow-up in width is unavoidable.

\begin{restatable}[$\RO$ powering non-closure]{maincor}{powering}
\label{cor:power-roabp}
    The following holds over fields of characteristic zero.
    There exists an $n$-variate polynomial $f$ computable by an $\RO$ of width $O(n)$ such that for any $d$, any $\RO$ computing $f^d$ requires width at least $\binom{d + n/2}{n/2}$ in every variable order.
\end{restatable}
\begin{remark*}
    We give two example polynomials to prove the hardness of powering for $\RO$. The first is the elementary symmetric polynomial (this lower bound will also prove \autoref{thm:sym-backward}) and the second is a quadratic polynomial inspired by the proof of \autoref{thm:roabp-non-closure}.
\end{remark*}

Another corollary of \autoref{thm:sym-backward} is that computing the resultant and the discriminant is hard for $\RO$.

\begin{restatable}[$\RO$ discriminant non-closure]{maincor}{discriminant}
\label{cor:discriminant-hard}
    The following holds over fields of characteristic zero. For all $n$, there exists an $n$-variate polynomial $f(\vecx,y)$ computable by an $\RO$ of width $O(n)$ such that any $\RO$ computing the discriminant $\operatorname{Disc}_y(f)$ requires width at least $2^{\Omega(n)}$ in every variable order.
\end{restatable}
\begin{remark*}
    As an immediate consequence of the corollary above, we get that $\RO$ is not closed under taking resultants.
\end{remark*}

\paragraph{Related Work.} There have been many lines of investigation into roABPs from the point of view of lower bounds~\cite{Nisan1991, KNS2020, AFSSV2018}, PIT algorithms~\cite{RS2005,BS2021,AGKS2015,GG2020, FS2013, FSS2014, AFSSV2018, AGKS2015, GKST2017, BG2022, ST2024}, border complexity~\cite{DDSdeborder21,DS2022, BIMPS2020, BDI2021}, algebraic meta-complexity~\cite{BDGT2025, BT2024} and so on. 

Our work is closely related to that of Kayal, Nair, and Saha~\cite{KNS2020}, who proved separations between the power of $\RO$s and multilinear depth-$3$ circuits. Non-closure results of a similar flavour to ours have also been proved by Saha and Thankey~\cite[Appendix E.1]{ST2024}. They construct explicit families of polynomials that require $\RO$ of exponential size, but arise from applying invertible linear transformations to sparse polynomials $f$ that have linear $\RO$ complexity. Some of their ideas, such as those involving the use of expander graphs, also appear in our work.

Similar separations between roABPs and other models (such as read-twice ABPs) were also addressed in the work of Anderson, Forbes, Saptharishi, Shpilka and Volk~\cite{AFSSV2018}. We re-prove a result from this work separating depth-$2$ algebraic circuits (products of linear polynomials) from roABPs in order to understand the roABP complexity of some explicit symmetric functions. Our lower bound is proved for the specific case of the determinant of a \emph{Circulant matrix}, which is a naturally occurring mathematical object and hence may be independently interesting.

\subsection{Proof Techniques} 

The main technique for understanding the roABP complexity of a polynomial is a characterization due to Nisan~\cite{Nisan1991}, who showed that the roABP complexity of a polynomial $f$ (or more precisely the \emph{width} of the smallest roABP computing $f$) in a given order is captured by the ranks of certain matrices related to $f$, also known as the \emph{evaluation dimension of $f$} (formally defined in \autoref{def:eval-dim}). We also heavily rely on this notion in our work.

\paragraph{Factor non-closure.} To construct our examples of polynomials that are efficiently computable by roABPs but hard to factor, we start with an analogous construction for a weaker setting, that of \emph{sparse polynomials.} The following is a well-known construction due to~\cite[Example 5.1]{GK1985}. 
\[
f(x_1,\ldots, x_n) \;=\; \prod_{i=1}^n(x_i^d - 1) = \prod_{i=1}^n (x_i-1) \cdot \underbrace{\prod_{i=1}^n(1+x_i+\cdots + x_i^{d-1})}_{g(x_1,\ldots, x_n)}.
\]
Note that the polynomial $f$ has $2^n$ monomials while its factor $g$ has $d^n$ monomials. This thus yields an example of a polynomial whose factors have many more monomials than the polynomial itself.

We would like to extend this to the setting of $\RO$. Unfortunately, the example above does not work as is, as the polynomial $g$ is a product of univariate polynomials and hence has a small roABP. Our idea is to `lift' this sparsity lower bound to an roABP lower bound. 

The basic idea of lifting, which has proven powerful in the area of Boolean complexity~\cite{RGR2022} and also Algebraic Proof complexity~\cite{FSTW2021}, is to start with a function $f$ that is hard for a simpler computational model (in this case sparse polynomials) and convert it to a function $g$ that is difficult for a much more powerful model by replacing the variables of $f$ by functions (typically called `gadgets') in a small number of new variables to obtain $g.$ A version of this idea can be used to lift degree lower bounds on the multilinear representation for some functions to lower bounds for algebraic proof systems based on roABPs~\cite{FSTW2021}.

Inspired by~\cite{FSTW2021}, we replace the variables of the polynomial $f$ by quadratic multilinear monomials in a new set of variables $y_1,\ldots, y_m$ where $m = \Omega(n).$ We can associate this replacement with an undirected graph $G$ on $m$ vertices and $n$ edges. We show that, \emph{as long as $G$ is a sufficiently good constant-degree expander}, the corresponding `lifted' polynomials $f_G$ and $g_G$ are easy and hard respectively for roABPs with similar parameters to the case of sparse polynomials. 

The crucial property of expander graphs that allows us to prove a lower bound on $g_G$ is the Expander Mixing lemma.
It can be used to show that given any balanced partition of the vertices of $G$, there is a large induced matching between the two sets in the parts. This allows us to find a large identity matrix as a submatrix of the evaluation matrix of $g_G$, leading to strong bounds on its evaluation dimension.

\paragraph{The complexity of powering.} We give two examples to demonstrate that roABPs are not closed under powering.

The first is a quadratic polynomial $g$ whose monomials again correspond to a constant-degree expander graph as in the previous result. The Expander Mixing lemma can again be used to argue that large powers of $g$ have large evaluation dimension.

The second example is just an elementary symmetric polynomial. Symmetric polynomials are particularly natural to study in the setting of roABPs, since the polynomials have the same complexity under any variable ordering. In particular, studying the complexity of a symmetric polynomial turns into understanding the ranks of combinatorially defined matrices. In the setting of a power of the elementary symmetric polynomial, we are able to show that this matrix has large rank.

\paragraph{No Bl\"{a}ser-Jindal type results for roABPs.} Already the example of the elementary symmetric polynomial above shows that for the simple polynomial $f(y_1,\ldots, y_n) = y_i^d$, the symmetric polynomial $f(\ESym^1_n,\ldots, \ESym^n_n)$ is hard to compute for roABPs. 

To prove a converse result, we use the symmetric polynomial ${f_{\mathrm{sym}} \coloneqq \sum_{d=0}^n \ESym^d_n(x_1^k,\ldots, x_n^k)}$. In this case, we need to understand the complexity of the polynomial $f$ (such as \(f(\ESym^1_n,\ldots,\ESym^n_n) = f_{\mathrm{sym}}\)). It turns out that the polynomial $f$ in this case is completely understood~\cite{achille2021} and is closely related to the determinant of the Circulant matrix. To prove the lower bound, we prove an roABP lower bound on this determinant, which we believe is independently interesting.

\paragraph{Outline.} We begin with preliminaries in Section~\ref{sec:prelims}. Section~\ref{sec:factor-non-closure} contains the proof of \autoref{thm:roabp-non-closure}. In Section~\ref{sec:sym-poly-complexity}, we prove \autoref{thm:sym-forward} and \autoref{thm:sym-backward}. Finally, in Section~\ref{sec:corollaries}, we prove all the corollaries.

\section{Notations and Preliminaries}
\label{sec:prelims}

Throughout the paper, we will use a growing parameter $n > 0$ to denote the number of variables in the polynomial. Let $\vecx = \inparen{x_1, \dots, x_n}$ be the set of indeterminates. A monomial of the form $x_1^{e_1} \cdots x_n^{e_n}$ is denoted as $\vecx^\vece$, where $\vece = (e_1, \ldots, e_n) \in \N^n$. The degree of a monomial $\vecx^\vece$ is defined as $\deg(\vecx^\vece) \coloneqq e_1 + \cdots + e_n$. 
The degree of a polynomial $f \in \F[x_1,\dots, x_n]$ is defined as the maximum degree of its constituent monomials.
We use $\coef_{\vecx^\vece}(f)$ to denote the coefficient of the monomial $\vecx^\vece$ in $f$.

\subsection{Read-Once Oblivious Algebraic Branching Programs (\texorpdfstring{$\RO$}{roABP})}

Our model of interest arises as a natural restriction of Algebraic Branching Programs (ABPs), which we describe next.
An \emph{Algebraic Branching Program} (ABP) is a layered and directed graph with a source vertex $s$ and a sink vertex $t$. All edges connect vertices from layer $i$ to $i + 1$. Further, the edges are labeled with affine polynomials over the underlying field $\F$. For every path $\gamma$ from $s$ to $t$, $\operatorname{wt}( \gamma )$ is the product of labels on the edges of the path $\gamma$. The polynomial computed by the ABP is defined as $$f \coloneqq \sum_{\text{path } \gamma: s \rightsquigarrow t} \operatorname{wt}( \gamma ).$$

The \emph{depth} of an $\ABP$ is defined as the number of layers in the graph, and the \emph{width} is the maximum number of nodes in a layer across the graph. The number of vertices used in the graph is the \emph{size} of the $\ABP$. 
The $\RO$ model is a restriction of $\ABP$s, which we define below.
\begin{definition}[$\RO$] \label{def:roabp}
    Let $n \in \N$ be arbitrary and fix a permutation $\pi: [n] \to [n]$. An $\RO$ in the order $\pi$ computing an $n$-variate polynomial $f(\vecx)$ is an $\ABP$ where in the $i$-th layer the edge labels are univariate polynomials  over $x_{\pi(i)}$. 
    
    The size of an $\RO$ is defined as the number of vertices it contains, and the width is defined as the maximum number of vertices in any layer.
\end{definition}

In a foundational work, Nisan~\cite{Nisan1991} characterized the complexity of an $\ABP$ in the non-commutative setting with the rank of certain matrices. Remarkably, the characterization extends to $\RO$ as well. We define the relevant matrix to formally state this characterization.

\begin{definition}[Nisan Matrix]
\label{def:nisan-matrix}
    Consider an $n$-variate polynomial $f(\vecx)$ and a variable partition $Y \sqcup Z = \{x_1, \ldots, x_n\}$. The \emph{Nisan matrix} of $f$ with respect to $Y,Z$, denoted as $M_{Y,Z}(f)$, is the matrix whose rows are indexed by monomials $m_Y$ over $Y$ and whose columns are indexed by monomials $m_Z$ over $Z$. Its entry at $(m_Y, m_Z)$ is defined as $$M_{Y,Z}(f)\bigg[m_Y, m_Z\bigg] \; = \; \coef_{m_Y \cdot m_Z}(f).$$
\end{definition}

Historically, the Nisan Matrix has also been referred to as the coefficient matrix or partial derivative matrix. The width of an $\RO$ computing a polynomial $f$ can be exactly characterized by the rank of the Nisan matrix of $f$~\cite[Lemma 4.5.8]{Forbes2014}.

\begin{theorem}[$\RO$ characterization]
\label{thm:roabp-char}
    Let $f(\vecx)$ be an $n$-variate polynomial, and fix a permutation $\pi$ on variables. 
    For each $i \in [n]$, consider the partition $Y_i \coloneqq \{\pi(x_1), \dots, \pi(x_i)\}$ and $Z_i = \{\pi(x_{i+1}), \dots, \pi(x_n)\}$, and let $M_{Y_i,Z_i}(f)$ denote the corresponding Nisan matrix. 

    \noindent The width of the smallest $\RO$ computing $f$ in the order $\pi$ is exactly $\max_{i \in [n]} \rank(M_{Y_i,Z_i}(f))$. Moreover, the size of the smallest $\RO$ is exactly $\sum_{i \in [n]} \rank(M_{Y_i,Z_i}(f))$.
\end{theorem}

We next prove a lemma to demonstrate the usefulness of the $\RO$ characterization, which will be used in our later proofs.

\begin{observation}
\label{ob:rochar-example}
    Consider an $n$-variate polynomial as follows 
    $$ f \; \coloneqq \; \prod_{i \in [n]} \left(1 + x_i + x_i^2 + \ldots + x_i^{d-1}\right).$$ Let $Y = \{y_1, \ldots, y_n\}$ and $Z = \{z_1, \ldots, z_n\}$ be disjoint set of variables. 
    Define a $2n$-variate polynomial $\tilde{f} \coloneqq f(y_1 z_1, \dots, y_n z_n)$.
    The rank of the Nisan matrix $M_{Y,Z}(\tilde{f})$ is $d^n$.
\end{observation}
\begin{remark}
    Note that $f$ itself can be computed by a constant width $\RO$ in any order.
\end{remark}
\begin{proof}
   Observe that for every monomial $m_Y$ over $Y$ such that each variable has degree at most $d-1$ in $m_Y$, there is a unique monomial $m_Z$ of the same form over $Z$ such that $\coef_{m_Y \cdot m_Z}(\tilde{f})$ is not zero, and reciprocally, for any monomial $m_Z$ over $Z$ there is a unique monomial $m_Y$. Consequently, the Nisan matrix $M_{Y,Z}(\tilde{f})$ is a permutation matrix, and hence has rank $d^n$.
\end{proof}

\paragraph{Evaluation Dimension.}
An alternative perspective on the Nisan matrix was introduced by Saptharishi~\cite[Section 6]{FS2013}. As we will see in our proofs, this viewpoint often makes it easier to reason about $\RO$ complexity.

\begin{definition}[Evaluation Dimension]
\label{def:eval-dim}
    Let $f(\vecx)$ be an $n$-variate polynomial on $X = \{x_1, \ldots, x_n\}$ over a field $\F$, and a subset of variables $Y$ and $Z \coloneq X \backslash Y$. The \emph{evaluation dimension} of $f$ with respect to the partition $Y \sqcup Z$ is defined as
    $$
    \evalDim_{Y,Z}(f) \;\coloneqq\; \rank \left( \setdef{f(Y, \veca)}{\veca \in \F^{|Z|}} \right).
    $$
\end{definition}

Over large fields, the evaluation dimension is equivalent to the rank of the Nisan matrix. However, this equivalence does not hold when restricting the evaluation points, e.g. to the Boolean cube. Nevertheless, the evaluation dimension is always a lower bound of the rank of the Nisan matrix (\cite[Lemma 11.9]{Sap2021}, and see also \cite[Corollary 4.5.12]{Forbes2014}).

\begin{theorem}
\label{thm:evaldim-ro}
    Let $f(\vecx)$ be an $n$-variate polynomial, and fix a permutation $\pi$ on variables. For a variable partition $Y \sqcup Z$ with $Y = \{x_{\pi(1)},\ldots,x_{\pi(i)}\}$ and $Z = \{x_{\pi(i+1)},\ldots,x_{\pi(n)}\}$, 
    any $\RO$ that computed $f$ in the order $\pi$ has width at least $\evalDim_{Y,Z}(f)$.

    Conversely, if the field \(\F\) is infinite, there is a \(\RO\) computing \(f\) of width $\evalDim_{Y,Z}(f)$.
\end{theorem}

\subsection{Elementary Symmetric Polynomials}

Symmetric polynomials are those that are invariant under any permutation of the variables. A fundamental and well-studied family within symmetric polynomials is the elementary symmetric polynomials, which are defined as follows.

\begin{definition}[Elementary Symmetric Polynomial] \label{def:ESym}
    The \emph{elementary symmetric polynomial} of degree $d$, on variables $x_1, \dots, x_n$, is defined as
    $$
    \ESym_n^d(x_1, \ldots, x_n) \;\coloneqq\; \sum_{1 \leq i_1 < \ldots < i_d \leq n} x_{i_1} \dots x_{i_d}.
    $$
\end{definition}

Whenever clear from the context, we write $e_d \coloneqq \ESym_n^d$ to denote the degree-$d$ elementary symmetric polynomial in $n$ variables. A more convenient way to define these polynomials is via the following generating functions:
\begin{equation}
    \prod_{i \in [n]} \left(1 + x_i \cdot t \right) \; = \; \sum_{i =0}^{n} \ESym_n^i(\vecx) \cdot t^i. \label{eq:esym-product}
\end{equation}

These polynomials are called \emph{elementary} because they form the fundamental building blocks for all symmetric polynomials.
For any $n$-variate polynomial $f(\vecx)$, we define the $n$-variate symmetric polynomial $f_{\sym} \coloneqq f(e_1, \dots, e_n)$.

\begin{theorem}[Fundamental Theorem of Symmetric Polynomials] \label{thm:fundamental-sym}
    Let $R$ be any commutative ring, and let $g \in R[x_1, \ldots, x_n]$ be a symmetric polynomial. Then there exists a unique polynomial $f \in R[y_1, \ldots, y_n]$ such that
    $$
    g \;=\; f_{\sym} \;\coloneqq\; f(\ESym_n^1, \ldots, \ESym_n^n).
    $$
\end{theorem}

We refer to \cite[Theorem IV.6.1]{Lang} for the proof of Fundamental Theorem of Symmetric Polynomials (see also \cite{BC2017}).
We will also need the following variable partitioning lemma, which is a special case of \cite[Theorem 1.1]{Merca2013}.
\begin{lemma}[$\ESym$ Variable Partition]\label{lem:esym-varpartition}
Let $Y \sqcup Z$ be a partition of the variables. 
Then,
$$
    \ESym_{\lvert Y \sqcup Z \rvert}^d(Y,Z) \; = \; \sum_{i=0}^{d} \left( \ESym_{|Y|}^i(Y) \cdot \ESym_{|Z|}^{d-i}(Z)\right).
$$
\end{lemma}
\begin{proof}
   Let $Y = \{y_1, \ldots, y_m\}$ and $Z = \{z_1, \ldots, z_n\}$. Then using \autoref{eq:esym-product} we can write, 
   \begin{align*}
       \sum_{i =0}^{m}\, \ESym_{m}^i(Y) \cdot t^i &\;=\; \prod_{i=1}^{m}\, (1 + y_i\cdot t)\\
       \sum_{i =0}^{n}\, \ESym_{n}^i(Z) \cdot t^i &\;=\; \prod_{i=1}^{n}\, (1 + z_i\cdot t)
   \end{align*}
   Taking the product of the two polynomials above, and comparing the coefficients of $t^d$ on both sides proves the lemma.
\end{proof}

As a direct consequence of extracting the coefficient of $t^d$ from \autoref{eq:esym-product}, Shpilka and Wigderson~\cite{SW2001} (crediting Ben-Or) presented the following identity for elementary symmetric polynomials, which yields a near-optimal $\RO$ of width $O(n)$ computing $\ESym_n^d$ in any variable order:

\begin{proposition}[\protect{\cite[Theorem 5.1]{SW2001}}]\label{prop:esym-produni}
    For any $n \in \mathbb{N}$ and $d \leq n$, let $\omega$ be a primitive $n$-th root of unity. There exist $\beta_0, \ldots, \beta_{n-1} \in \C$ such that $$\ESym_n^d(x_1, \dots, x_n) \; =  \; \sum_{0 \leq j < n} \beta_j \left(1 + \omega^j x_1\right) \cdot \left(1 + \omega^j x_2\right) \cdots \left(1 + \omega^j x_n\right).$$
\end{proposition}

\begin{remark}\label{rem:esym-roABP}
    $\ESym_n^d$ can also be computed by a provably tight $\RO$ of width $\min(d+1,n-d+1)$ in any variable order using only coefficients \(0\) and \(1\) (see \cite[Construction 1.2]{RT2022}).
\end{remark}

\subsection{Resultant and Discriminant}

We recall the definitions and properties of resultant and discriminant from factorization literature. We encourage readers to refer to~\cite[Chapter 6]{GG2013} for a more detailed textbook treatment of these concepts.

\begin{definition}[Resultant] 
\label{def:resultant}
    Consider two $n$-variate polynomials $f, g \in \F[\vecx][y]$ as follows:
    $$f \;\coloneqq\; \sum_{i = 0}^{d_1} f_i(\vecx) \cdot y^i \quad \text{and} \quad g \;\coloneqq\; \sum_{i=0}^{d_2} g_i(\vecx) \cdot y^i.$$ Define the \emph{Sylvester matrix} of $f$ and $g$ as the following $(d_1 + d_2) \times (d_1 + d_2)$ matrix:
    \begin{align*}
        \mathbf{S}_y(f, g)\; =\;
        \begin{pmatrix}
            f_{d_1} &   &   &  & g_{d_2} &  &  &  &  \\
            f_{d_1-1} & f_{d_1} &   &  & g_{d_2 - 1} & g_{d_2} &  &  &  \\
            \vdots & f_{d_1-1} & \ddots &  & \vdots & g_{d_2-1} & \ddots &  &  \\
            \vdots & \vdots & \ddots &  f_{d_1} & \vdots & \vdots & \ddots & g_{d_2} &  \\
            f_0 & \vdots &   & f_{d_1 - 1} &  g_0 & \vdots &   & g_{d_2-1} \\
             & f_{0} & \ddots & \vdots &  & g_0 & \ddots &  \vdots \\
             &  & \ddots & \vdots &  &  & \ddots & \vdots \\
             &  &  & f_{0} &  &  &  & g_0
        \end{pmatrix}
    \end{align*}
    Then the resultant of the two polynomials with respect to $y$ is defined as the determinant of the Sylvester matrix as: $$\operatorname{Res}_y(f,g) \;\coloneqq \;\Det(\mathbf{S}_y(f, g)).$$
\end{definition}

The resultant of two polynomials is non-zero if and only if their $\gcd$ is $1$. A well-known case of resultant relevant for factoring algorithms is the discriminant.

\begin{definition}[Discriminant]
\label{def:discriminant}
    Consider a $n$-variate polynomial $f$. The \emph{discriminant} with respect to \(y\) of $f$ is defined as the resultant, with respect to \(y\), of $f$ and its \(y\)-derivative, i.e., $$\operatorname{Disc}_y(f) \;\coloneqq\; \operatorname{Res}_y(f, \partial_y f).$$
\end{definition}

The following well-known observation will be useful in the analyses of complexity of the resultant and the discriminant.

\begin{observation}[see \protect{\cite[Chapter 3]{CLO2005}}]
\label{ob:res-product}
    Let $f = \prod_{i \in [n]} (y - \alpha_i)$ and $g = \prod_{i \in [m]} (y - \beta_i)$ be two univariate polynomials. Then the resultant of $f$ and $g$ with respect to $y$ is given by
    $$ \operatorname{Res}_y(f, g) \;=\; \prod_{i \in [n]} g(\alpha_i).$$
\end{observation}

\section{\texorpdfstring{$\RO$}{roABP} Factor Non-Closure}
\label{sec:factor-non-closure}

To prove Theorem \ref{thm:roabp-non-closure}, we need a polynomial of low $\RO$ complexity, that has a factor of high $\RO$ complexity.
We will use explicit expander graphs for this purpose.
The only property we require from the expander graph is that, for any sufficiently large partition of its vertex set into two parts, it contains a large induced matching between the two parts.

\begin{lemma}[Induced Matching Lemma] \label{lem:expander-matching}
    %Fix $\epsilon \in (0,1)$ with $\epsilon \in [ \tfrac{1}{3},\tfrac{2}{3}]$. 
    For every $n \in \N$ there exists a constant degree graph $G_n = (V,E)$ on $n$ vertices such that the following holds: for any partition $(S, T)$ of $V$ with $|S| = \epsilon n$ and $|T| = (1 - \epsilon) n$ where \(\epsilon \in \left[ \tfrac{1}{3},\tfrac{2}{3}\right]\), the graph contains $\Omega(n)$ edges between $S$ and $T$ that form an induced matching.
\end{lemma}

\begin{proof}[Proof Sketch]
    There exists an absolute constant $\delta \in (0,1)$ such that, for any $k \in \N$ with $k \geq 1$, we can construct explicit $k$-regular expander graphs $G_n = (V,E)$ whose second-largest eigenvalue is at most $k^\delta$; see~\cite{RVW2000}. 
    
    When $k$ is chosen to be sufficiently large such that the second-largest eigenvalue of $G_n$ is strictly smaller than $k / 3$, then the lemma follows as an easy consequence of the Expander Mixing Lemma~\cite{AlonC06} (see also \cite[Lemma 2.5]{HLW06}). See, for example, \cite[Claim 4]{jukna2008}.
\end{proof}

Define an $n$-variate polynomial $P_{G}$ associated with constant degree graph $G_n = (V,E)$ guaranteed by \autoref{lem:expander-matching} as follows: 
\begin{equation}
    P_G ~ \coloneqq ~ \prod_{(i,j) \in E} \inparen{ (x_ix_j)^d - 1}. \label{eq:poly-expander}
\end{equation} 
Since the degree of the graph is constant, the sparsity of $P_G$ is $2^{|E|} = 2^{O(n)} =: w$. 
Therefore, $P_G$ can be computed by an $\RO$ of width $w$ in every variable order.
To prove the hardness of its factor, consider the following polynomial $Q_G$: 
\begin{equation}
    Q_G ~ \coloneqq ~ \prod_{(i,j) \in E} \inparen{1 + \inparen{x_ix_j} + \inparen{x_ix_j}^2 + \ldots + \inparen{x_ix_j}^{d-1}}. \label{eq:poly-expander-factor}
\end{equation}
It is well known that $(1+ x + x^2 + \ldots + x^{d-1})(x-1) = x^d - 1$. 
Using the identity, we immediately obtain $$P_G \;=\; Q_G \cdot \prod_{(i,j) \in E} \left(x_i x_j - 1\right).$$

\begin{lemma}\label{lem:factor-roabp}
    The polynomial $Q_G$ defined in \autoref{eq:poly-expander-factor} requires an $\RO$ of width $d^{\Omega(n)}$ in every variable order.
\end{lemma}
\begin{proof}
    Let $\pi$ be any variable order on the variables $X = \{x_1, \ldots, x_n\}$, and consider the partition $Y = \{x_{\pi(1)}, \ldots, x_{\pi(n/2)}\}$ and $Z = X \setminus Y$.

    Let $Y$ and $Z$ also denote the partition of vertices of $G_n$. 
    Then from \autoref{lem:expander-matching}, we know there exists an induced matching $\calM$ between $Y$ and $Z$ of size $\Omega(n)$. Define  
    \begin{align*}
        \tilde{f} \; &\coloneqq \; \prod_{(i,j) \in \calM} \inparen{1 + \inparen{x_ix_j} + \inparen{x_ix_j}^2 + \ldots + \inparen{x_ix_j}^{d-1}} \\
        \; &=\; \; \prod_{i \in [t]} \inparen{1 + \inparen{y_{i}z_{i}} + \inparen{y_{i}z_{i}}^2 + \ldots + \inparen{y_{i}z_{i}}^{d-1}},
    \end{align*}
    where for every $i \in [t]$, $y_i$ is a variable in $Y$ and $z_i$ is a variable in $Z$, and $t = \Omega(n)$. Here we have used the fact that $\calM$ is an induced matching. In particular, $\tilde f$ is obtained from $Q_G$ by setting to zero the variables which are not in the matching $\calM$. Hence, the rank of the Nisan matrix can only decrease. 
    Finally, by \autoref{ob:rochar-example},
    $$\rank\left(M_{Y,Z}(Q_G)\right) \;\geq\; \rank\left(M_{Y,Z}\left(\tilde{f}\right)\right) \;\geq\; d^{\Omega(n)}.$$
    We obtain the claimed lower bound for width of $\RO$ computing $Q_G$ by \autoref{thm:roabp-char}.
\end{proof}

We will now use the discussion so far to give the complete proof of the factor non-closure result.

\roabpnonclosure*
\begin{proof}
    Consider an $n$-variate polynomial $g \coloneqq Q_{G_{n-1}} + z$, where $z$ is an auxiliary variable and $Q_{G_{n-1}}$ is defined as in \autoref{eq:poly-expander-factor} using a constant-degree graph $G_{n-1}$.
    We then define
    $$ f \;\coloneqq\; g \cdot \prod_{(i,j) \in E} \left(x_i \cdot x_j - 1\right) \; = \; P_G + z \cdot \prod_{(i,j) \in E} \left(x_i \cdot x_j - 1\right).$$

    As argued after \autoref{eq:poly-expander}, both $P_G$ and $\prod_{(i,j) \in E} (x_i \cdot x_j - 1)$ have sparsity $2^{O(n)}$ and hence we can compute them by an $\RO$ of width $w = 2^{O(n)}$ in every variable order. Therefore, $f$ itself admits an $\RO$ of width $w$ in every variable order.

    Observe that $g$ is an irreducible polynomial because it is linear in the auxiliary variable $z$.\footnote{See \cite[Example 5.1]{GK1985} where the hardness is lifted to irreducible factor by considering $g = Q_G + n$.}
    Further, by \autoref{lem:factor-roabp}, any $\RO$ computing $Q_G + z$ must have width at least $d^{\Omega(n)}$ in every variable order. 
    Since $d \geq n$, the claimed width lower bound for $\RO$ computing $g$ follows.
\end{proof}

\section{\texorpdfstring{$\RO$}{RO} Complexity of Symmetric Polynomials}
\label{sec:sym-poly-complexity}

In the following two sections we prove \autoref{thm:sym-forward} and \autoref{thm:sym-backward} along with their corollaries. 

\subsection{\texorpdfstring{$f_\sym$}{Sym Poly} is easy, but \texorpdfstring{$f$}{f} is hard}

In this section, we work over fields of characteristic zero. 
For the proof of \autoref{thm:sym-forward}, we consider $f_{\sym} = \sum_{d=0}^n \ESym_n^d(x_1^k, \ldots, x_n^k)$ for a suitable choice of $k \in [n]$ to be fixed later.

Let us consider the polynomial 
\begin{align*}
    g(y_1,\ldots,y_n,t,z_0,\ldots,z_{k-1})  \; \coloneqq \; \prod_{j = 0}^{k-1} \left( 1 + \sum_{i \in [n]} y_i \cdot \left(t\cdot z_j\right)^i \right).
\end{align*}
The polynomial \(g\) is symmetric in the variables \(z_0,\ldots,z_{k-1}\). So by Theorem~\ref{thm:fundamental-sym}, there exists a polynomial \(\tilde{g} \in \Z[y_1,\ldots,y_n,t,z_0,\ldots,z_{k-1}]\) such that \(g(\vecy,t,\vecz) = \tilde{g}(\vecy,t,e_1(\vecz),\ldots,e_k(\vecz))\). Notice that if \(\omega\) is a \(k\)-th primitive root of the unity, \autoref{eq:esym-product} implies 
\begin{align*}
    \begin{cases}
        e_i(\omega^0,\ldots,\omega^{k-1}) = 0 & \text{for }1\le i<k, \\
        e_k(\omega^0,\ldots,\omega^{k-1}) = 1. &
    \end{cases}
\end{align*}
Let us define 
\begin{equation}
    f(y_1,\ldots,y_n) \;\coloneqq\; \tilde{g}(\vecy,1,0,\ldots,0,1) \in \Z[\vecy].\label{eq:circulant-poly}
\end{equation}
The previous paragraph ensures that for any \(k\)-th primitive root of the unity \(\omega\), we have
\begin{equation}
        f(y_1, \dots, y_n) \;=\;  \prod_{j = 0}^{k-1} \left( 1 + \sum_{i \in [n]} y_i \cdot \omega^{j\cdot i} \right). \label{eq:circulant-poly-eq}
\end{equation}

The following lemma shows that \(f\) is indeed the unique polynomial inducing the symmetric polynomial \(\sum_{d=0}^n \ESym_n^d(x_1^k, \ldots, x_n^k)\). 
The following is an argument in~\cite{achille2021}, which we reproduce here for completeness.

\begin{lemma*}[Circulant Polynomial]
    For any $n\in \N$ and odd positive integer $k \leq n$, 
    $$f_\sym(\vecx) 
    \;\coloneqq\; f(e_1(\vecx), \ldots, e_n(\vecx))
    \;=\; \sum_{d=0}^{n} \, \ESym_n^d\left(x_1^k, \dots, x_n^k\right).$$
\end{lemma*}

\begin{proof}
    Let \(\omega\) be a \(k\)-th primitive root of the unity. 
    Using the factorization identity ${(1 - t^k) = \prod_j (1 - \omega^j \cdot t)}$ together with \autoref{eq:esym-product} we obtain the following:
    \begin{align}
        \prod_{i \in [n]} \left(1 - x_i^k \cdot (-t)^k \right) \; &= \; \prod_{i =1}^n \, \prod_{j = 0}^{k-1} \, \left( 1 + \omega^j \cdot x_i \cdot t \right)\notag\\
         \; &= \; \prod_{j = 0}^{k-1} \, \left( 1 + \sum_{i =1}^{n} e_i(\vecx) \cdot (\omega^j \cdot t)^i\right).\notag
    \end{align}
    By \autoref{eq:circulant-poly-eq} and instantiating \(t\) by \(1\), we obtain
    \begin{equation*}
        \prod_{i \in [n]} \left(1 + x_i^k \right) \; = \; f \left(e_1(\vecx), \ldots, e_n(\vecx)\right). \qedhere
    \end{equation*}
    
\end{proof}

% \begin{claim}
% \label{clm:circulant-poly}
%     Let $h$ be the polynomial defined in \autoref{eq:circulant-poly}. 
%     For all $i \in [nk]$, define $h_i(\vecy) \coloneqq \coef_{t^i}\,[h(\vecy, t)]$. 
%     If $h_i \neq 0$ then $k$ divides $i$.
% \end{claim}
% \begin{proof}
%     Observe that $h(\vecy, \xi t) = h(\vecy, t)$ for every primitive $k$-th root of unity $\xi$. 
%     Comparing coefficients of $t^i$ on both sides give $h_i \cdot (1 - \xi^i) = 0$ for all $i$. 
%     Hence, if $h_i \neq 0$, we must have $\xi^i = 1$, which implies $k \mid i$.
% \end{proof}

We call the polynomial $f$ a circulant polynomial because it is closely related to the determinant of a Circulant matrix.\footnote{Specifically, in the case $k=n$, the homogeneous component of degree $k$ of the polynomial \(f\) is exactly the determinant of the circulant matrix of first row $(x_1,\ldots,x_n)$.}

In the following lemma, we show that the polynomial $f$ is hard for $\RO$ in every variable order over any field \(\F\) of characteristic \(0\).
\begin{lemma} 
\label{lem:sym-poly-evaldim}
    For any prime $k$ with $2 \leq k \leq n$, the $n$-variate polynomial $f$ defined in \autoref{eq:circulant-poly} requires $\RO$ width at least $2^{(k-1)/2}$ in any variable order.
\end{lemma}
\begin{proof}
    By instantiating a variable, the roABP width can only decrease. So it is sufficient to consider \(f'(y_1,\ldots,y_k) \coloneqq f(y_1,\ldots,y_k,0,\ldots,0)\).
    
    By the standard evaluation-dimension lower bound for $\RO$, it suffices to show the following. For any variable order $\pi$ and the variable partition $(U,V)$ where $U=\{y_{\pi(1)},\dots,y_{\pi( (k-1)/2)}\}$ and $V=\{y_{\pi( (k-1)/2+1)},\dots,y_{\pi(k)}\}$ we have 
    \[
    \evalDim_{U,V}(f') \;=\; \rank \left( \setdef{f'(\vecu, \veca)}{\veca \in \F^{|V|}} \right)\ge\;2^{\Omega(k)}.
    \]

    Re-writing \autoref{eq:circulant-poly-eq} in terms of variable partition $(U,V)$, we have 
    \begin{equation}
        f'(\vecu, \veca) \; = \; \prod_{j = 0}^{k-1} \bigg( \ell_j(\vecu) + \ell'_j(\veca) + 1 \bigg), \label{eq:circulant-var-split}
    \end{equation}
    where $\ell_j(\vecu)$ and $\ell'_j(\vecv)$ are linear polynomials in $U$ and $V$, respectively. 
    
    Arranging the coefficients in the linear forms $\{\ell_j(\vecu) + \ell'_j(\vecv)\}_{j}$ 
    as the rows of a $k\times k$ matrix yields a matrix $M$ whose $(j,i)$-th entry is $\omega^{j\cdot \pi(i)}$ for $j\in \{0,\ldots,k-1\}$ and $i \in [k]$. Note that $M$ can be obtained from the standard $k\times k$ DFT matrix $(\omega^{j\cdot i})_{i,j\in [k]}$ by permuting columns.

    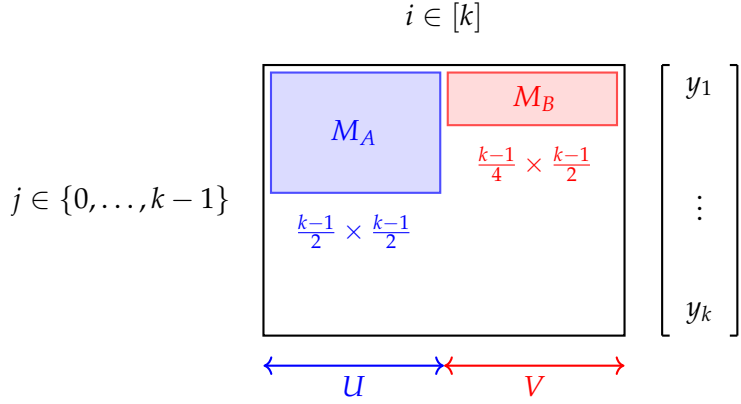
\begin{figure}[ht]
    \centering
    \begin{tikzpicture}
      % Parameters
      \pgfmathsetmacro{\k}{7}
      \pgfmathsetmacro{\matrows}{6}
      \pgfmathsetmacro{\matcols}{6}
      \pgfmathsetmacro{\cellw}{0.8}
      \pgfmathsetmacro{\cellh}{0.6}
      \pgfmathsetmacro{\belowpad}{1.0}
      \pgfmathsetmacro{\abovepad}{0.8}
      
      % Bounding box
      \path[use as bounding box]
        (0, -\belowpad) rectangle ({\matcols*\cellw}, {\matrows*\cellh + \abovepad});
    
      % Main matrix
      \draw[thick] (0,0) rectangle ({\matcols*\cellw}, {\matrows*\cellh});
    
      % Labels
      \node[left] at (-0.3, {\matrows*\cellh/2}) {$j \in \{0,\dots,k-1\}$};
      \node[above] at ({\matcols*\cellw/2}, {\matrows*\cellh + 0.3}) {$i \in [k]$};
    
      % Submatrix M_A
      \pgfmathsetmacro{\varows}{3}
      \pgfmathsetmacro{\vacols}{3}
      \pgfmathsetmacro{\gap}{0.1}
      \path[fill=blue!20, opacity=0.7, draw=blue, thick]
        ({\gap}, {\matrows*\cellh - \varows*\cellh + \gap})
        rectangle ({\vacols*\cellw - \gap/2}, {\matrows*\cellh - \gap});
      \node[blue, font=\bfseries] 
        at ({\vacols*\cellw/2}, {\matrows*\cellh - \varows*\cellh/2}) {$M_A$};
    
      % Submatrix M_B
      \pgfmathsetmacro{\vbrows}{1.5}
      \pgfmathsetmacro{\vbcols}{3}
      \path[fill=red!20, opacity=0.7, draw=red, thick]
        ({\vacols*\cellw + \gap/2}, {\matrows*\cellh - \vbrows*\cellh + \gap})
        rectangle ({(\vacols+\vbcols)*\cellw - \gap}, {\matrows*\cellh - \gap});
      \node[red, font=\bfseries] 
        at ({(\vacols+\vbcols/2)*\cellw}, {\matrows*\cellh - \vbrows*\cellh/2}) {$M_B$};
    
      % Vector y with brackets
      \pgfmathsetmacro{\vecx}{\matcols*\cellw + 0.5}
      \pgfmathsetmacro{\vecw}{1}
      \draw[thick] (\vecx,0) -- (\vecx,{\matrows*\cellh});
      \draw[thick] (\vecx+\vecw,0) -- (\vecx+\vecw,{\matrows*\cellh});
      % bracket tips
      \draw[thick] (\vecx,0) -- (\vecx+0.1,0);
      \draw[thick] (\vecx,{\matrows*\cellh}) -- (\vecx+0.1,{\matrows*\cellh});
      \draw[thick] (\vecx+\vecw-0.1,0) -- (\vecx+\vecw,0);
      \draw[thick] (\vecx+\vecw-0.1,{\matrows*\cellh}) -- (\vecx+\vecw,{\matrows*\cellh});
      % entries
      \node at ({\vecx+\vecw/2}, {\matrows*\cellh - 0.5*\cellh}) {$y_1$};
      \node at ({\vecx+\vecw/2}, {\matrows*\cellh/2}) {$\vdots$};
      \node at ({\vecx+\vecw/2}, {0.5*\cellh}) {$y_k$};
    
      % Arrows U, V
      \pgfmathsetmacro{\arrowy}{-0.4}
      \draw[<->, thick, blue] (0,\arrowy) -- (\vacols*\cellw,\arrowy);
      \node[below, blue] at (\vacols*\cellw/2,\arrowy) {$U$};
      \draw[<->, thick, red] (\vacols*\cellw,\arrowy) -- (\matcols*\cellw,\arrowy);
      \node[below, red] at ({(\vacols+(\matcols-\vacols)/2)*\cellw},\arrowy) {$V$};
    
      % Dimension labels
      \node[blue, below] 
        at (\vacols*\cellw/2, {\matrows*\cellh - \varows*\cellh - 0.05})
        {\small $\tfrac{k-1}{2} \times \tfrac{k-1}{2}$};
      \node[red, below] 
        at ({(\vacols+\vbcols/2)*\cellw}, {\matrows*\cellh - \vbrows*\cellh - 0.05})
        {\small $\tfrac{k-1}{4} \times \tfrac{k-1}{2}$};
    \end{tikzpicture}
    \caption{The matrix $M$ with $(j,i)$-th entry $\omega^{j\pi(i)}$, corresponding to the natural variable order. The highlighted submatrices $M_A$ and $M_B$ are used in the analysis of evaluation dimension.}
    \label{fig:matrix-structure}
    \end{figure}

    When $k$ is prime, Chebotarev’s theorem on roots of unity \cite{SL1996} states that every square submatrix of the DFT matrix (and hence also $M$) is nonsingular; see also \cite[Lemma1.3]{Tao2004} and \cite{Frenkel2004}.
    %\footnote{For a finite-field analogue, see \cite[Theorem3]{EK2025}.}.
    Since $|U| = (k-1)/2$, we can fix a subset $A \subseteq \{0,\ldots,k-1\}$ of size $(k-1)/2$ such that the set $\{\ell_j(\vecu)\}_{j \in A}$ corresponds to a square submatrix $M_A$ of $M$. Such a submatrix $M_A$ is non-singular due to Chebotarev's theorem.
    Consequently, the set $\{\ell_j(\vecu)\}_{j \in A}$ is linearly independent.
    We can assume that $\ell_i(\vecu)= u_i$ for $i \in A$, since an invertible linear transformation on the variables in $U$ does not change the evaluation dimension \(\evalDim_{U,V}(f')\).

    Consider any $B\subseteq A$. By Chebotarev's theorem, the submatrix $M_B$ with rows indexed by \(B\) and columns \(v_2,\ldots,v_{\lvert B\rvert +1}\) is invertible (the choice \(\lvert A \rvert = (k-1)/2\) ensures that \(V\) contains enough variables). So, for any \(b\in \F\), there is a unique point \(\beta_{B,b} \in \F^{\lvert B \rvert}\) such that for any \(j\) in \(B\), \(\ell'_j(b,\beta_{B,b},0,\ldots,0) + 1 = 0\).
    Similarly, for any other row index $\tilde{\jmath} \in \{0,\ldots,k-1\} \setminus B$, there exists a unique point \(\gamma_{B,\tilde{\jmath}} \in \F^{\lvert B\rvert +1}\)  such that for any \(j\) in \(B\cup\left\{\tilde{\jmath}\right\}\), we have \(\ell'_j(\gamma_{B,\tilde{\jmath}},0,\ldots,0) + 1 = 0\). It follows that \(\gamma_{B,\tilde{\jmath}}\) is of the form \((b_{\tilde{\jmath}},\beta_{B,b_{\tilde{\jmath}}})\) for a particular \(b_{\tilde{\jmath}} \in \F\).

    Since \(\F\) is infinite, we can choose \(b\) in \(\F\) outside of \(\setdef{b_{\tilde{\jmath}}}{\tilde{\jmath} \in \{0,\ldots,k-1\} \setminus B}\), and define \(\alpha_B \coloneqq (b,\beta_{B,b},0,\ldots,0)\). For any \(j\) in \(\{0,\ldots,k-1\}\):
    \begin{equation*}
        \ell'_j(\alpha_B)+1 = 0 \;\iff\; j \in B.
    \end{equation*}

    Consequently \(f'(\vecu,\alpha_B) \;=\; \left(\prod_{j\in B}\,u_j\right)\cdot\left(\prod_{j\in [k]\setminus B}(\ell_j(U) + c_{B,j})\right)\)
    where the \((c_{B,j})_{j \notin B}\) are non-zero constants.
    Since for each $B$, $f'(\vecu, \alpha_B)$ has a distinct lowest degree monomial \(\left(\prod_{j\in B}\,u_j\right)\), the set $\setdef{f'(\vecu, \alpha_B)}{B \subseteq A}$ is linearly independent.
    Therefore,
    \begin{align*}
        \evalDim_{U,V}\bigl(f')\bigr)\;&\ge\; \dim \setdef{f'(\vecu, \alpha_B)}{B \subseteq A} = 2^{(k-1)/2}.
    \end{align*}
    By the evaluation-dimension lower bound of \autoref{thm:evaldim-ro}, any $\RO$ computing $f$ (in any order) must have width at least $2^{(k-1)/2}$.
\end{proof}

\begin{remark}
    \label{rem:circulant}
    A close look at the above proof reveals that the lower bound also applies to the circulant polynomial $\prod_{j=0}^{k-1}(\sum_{i=1}^n y_i \omega^{ij})$ which is exactly the determinant of the circulant matrix. 
\end{remark}

The lower bound established in \autoref{lem:sym-poly-evaldim} serves as the key technical ingredient needed to prove \autoref{thm:sym-forward}.

\symforward*
\begin{proof}
    Fix $k$ to be a prime number between $n/2$ and $n$. We consider the symmetric polynomial $f_{\sym} \coloneqq \sum_{d = 0}^{n}\ESym_n^d(x_1^k, \dots, x_n^k)$.
    By applying \autoref{eq:esym-product} with each $x_i$ replaced by $x_i^k$, we obtain that $f_{\sym}$ admits an $\RO$ of constant width in every variable order.
    Moreover, by \autoref{lem:sym-poly-evaldim}, the width of any $\RO$ computing $f$ is at least $2^{(k-1)/2} = 2^{\Omega(n)}$.
\end{proof}

\subsection{\texorpdfstring{$f$}{f} is easy, but \texorpdfstring{$f_\sym$}{Sym Poly} is hard}
\label{sec:sym-backward}

To prove \autoref{thm:sym-backward}, it suffices to show the following technical lemma, which shows that taking powers of elementary symmetric polynomials is hard for $\RO$s. This lemma also implies \autoref{cor:power-roabp} from the introduction. In \autoref{sec:hard-power}, we will present an alternative proof of \autoref{thm:sym-backward} using a quadratic polynomial based on graph-based polynomial from~\autoref{sec:factor-non-closure}.

\begin{lemma}[Powers of $\ESym$] \label{lem:ESym-power}
    Let $k \leq n/2$. Any $\RO$ computing $\left(\ESym_n^k\right)^d$ in any variable order requires width at least $\binom{k + d}{k}$.
\end{lemma}

\symbackward*

\begin{proof}
    Let \(k = \lfloor n/2 \rfloor\).
    Consider the polynomial $f(x_1, \ldots, x_{n}) = x_k^k$. It is easy to see that $f$ can be computed by an $\RO$ of constant width.  However, by \autoref{lem:ESym-power}, any $\RO$ computing the symmetrisation
    $$
        f_{\sym} = f\big(\ESym^1_{n}, \ldots, \ESym^{n}_{n}\big) = \bigg(\ESym_{n}^k\bigg)^k
    $$
    must have width at least 
    \begin{equation*}
        \binom{2\lfloor n/2 \rfloor}{\lfloor n/2 \rfloor} = \Omega\left( 2^{n} / \sqrt{n} \right). \qedhere
    \end{equation*}
\end{proof}

\begin{remark}
    We recall that $\ESym_{n}^{\lfloor n/2 \rfloor}$ can be expressed as a sum of $n$ many products of univariate polynomials (see \autoref{prop:esym-produni}). Consequently, using the multinomial theorem, it follows that $(\ESym_{n}^{\lfloor n/2 \rfloor})^{\lfloor n/2 \rfloor}$ can be expressed as a sum of at most $O(2^{1.5n})$ many products of univariate polynomials. Hence, the bound we obtain in \autoref{thm:sym-backward} is almost optimal. 
\end{remark}
    
\begin{proof}[Proof of \autoref{lem:ESym-power}]
    Assume that $\left(e_k\right)^d$ is computed by an \(\RO\) of width \(w\) and variable order \(\pi\).
    Let $Y = \{x_{\pi(1)}, \ldots, x_{\pi(k)}\}$ and $Z = X \setminus Y$ be a partition of the variables $X$. By \autoref{thm:evaldim-ro}, we know that \[w \ge \evalDim_{Y,Z}\left((e_k)^d\right).\]

    Using \autoref{lem:esym-varpartition}, and the multinomial theorem, we can write the powers of the elementary symmetric polynomial $e_k$ as follows:
    \begin{align}
        \left(e_k(X)\right)^d &\; =\;  \left( \sum_{t=0}^k \, e_t(Y) \cdot e_{k-t}(Z) \right)^d \notag\\
                 &\;=\; \sum_{\substack{t_0 + \cdots + t_k = d \\ t_i \geq 0 }} \,\binom{d}{t_0,...,t_k} \, \bigg(e^{t_0}_{0}(Y) \cdots e^{t_k}_{k}(Y)\bigg) \cdot \bigg(e^{t_0}_{k-0}(Z) \cdots e^{t_k}_{k-k}(Z) \bigg).\label{eq:convolution}
    \end{align}

    To argue about the evaluation dimension of $(e_k(X))^d,$ we will need the following elementary fact from linear algebra.

    \begin{fact}
        \label{fac:lin-alg}
        If a matrix $M = \sum_{i=1}^r u_i v_i^T$ where $\{u_1,\ldots, u_r\}$ and $\{v_1,\ldots, v_r\}$ are linearly independent sets of vectors, then $M$ has rank exactly $r.$
    \end{fact}

    To use the above fact, we note that the algebraic independence  of the elementary symmetric polynomials (a consequence of \autoref{thm:fundamental-sym}) implies that the sets
    \[
    E = \{e_0^{t_0}(Y)\cdots e_k^{t_k}(Y): t_0+\cdots + t_k = d\} \text{ and } \tilde{E} = \{e_{k-0}^{t_0}(Z)\cdots e_{k-k}^{t_k}(Z): t_0+\cdots + t_k = d\}
    \]
    are both linearly independent sets of polynomials. Further, each term on the right-hand side of \autoref{eq:convolution} (corresponding to a tuple $(t_0,\ldots, t_k)$ summing to $d$) has an evaluation matrix  that is the outer product of the coefficient vectors of the corresponding polynomials in $E$ and $\tilde{E}$, scaled by a suitable multinomial coefficient (which is non-zero because we have assumed that the characteristic of the underlying field is $0$). This implies that the evaluation matrix of $(e_k(X))^d$ has rank exactly the number of terms which is $\binom{k+d}{k}.$
\end{proof}

\section{Non-closure corollaries for \texorpdfstring{$\RO$}{RO}}
\label{sec:corollaries}

We will now give the proofs of corollaries stated in \autoref{sec:main-results}. We use observations from the earlier sections to show that operations such as powering, computing resultant, and discriminant can be hard for $\RO$.

\subsection{Hardness of Powering: a second example}
\label{sec:hard-power}
In this section, we give a second proof of (a slightly weaker form of) \autoref{cor:power-roabp}. Note that we already proved this in the form of \autoref{lem:ESym-power}. By \autoref{prop:esym-produni} and the following remark, we know that $\ESym_{2n}^{n}$ admits an $\RO$ of width $O(n)$ in any variable order. On the other hand, any $\RO$ that computes $(\ESym_{2n}^{n})^d$ must have width at least $\binom{n + d}{n}$. 

Inspired by the graph-based polynomial which was used to prove factor non-closure in \autoref{sec:factor-non-closure}, we can even define a quadratic polynomial $Q$ and prove that powering $Q$ is hard for this polynomial $\RO$. The lower bound we obtain is slightly weaker, but the example is even simpler since $Q$ is just a quadratic polynomial, as opposed to the high-degree and high-sparsity elementary symmetric polynomial.

\begin{corollary}[Variant of \autoref{cor:power-roabp}]
\label{cor:power-roabp-expander}
    The following holds over fields of characteristic zero.
    There exists an $n$-variate quadratic polynomial $Q$ computable by an $\RO$ of width $O(n)$ such that for any $d$, any $\RO$ computing $Q^d$ requires width at least $\binom{d + m}{m}$ in every variable order where $m  = \Omega(n).$
\end{corollary}

\begin{proof}
    Let $G = (V,E)$ be a constant degree graph on $n$ vertices such that \autoref{lem:expander-matching} holds. Define the quadratic polynomial:
    \begin{align}
        Q_G \;=\; \sum_{(i,j) \in E} x_i x_j \label{eq:expander-poly-quadratic}
    \end{align}
    where variables $x_i$ correspond to the vertices of $G$. It is easy to observe that $Q_G$ can be computed by an $\RO$ of width $|E| = O(n)$ in any variable order. We will prove that any $\RO$ computing $Q_G^d$ must have large width.
    
    Let $\pi$ be any variable order on the variables $X = \{x_1, \ldots, x_n\}$, and consider the partition $Y = \{x_{\pi(1)}, \ldots, x_{\pi(n/2)}\}$ and $Z = X \setminus Y$. Let $Y$ and $Z$ also denote the partition of vertices on $G$. By \autoref{lem:expander-matching}, there exists an induced matching $\calM$ between $Y$ and $Z$ of size $\Omega(n)$. By renaming the variables if necessary we assume that the matching is between the vertices corresponding to $y_i$ and $z_i$ where $i\in [t]$  and $t = \Omega(n)$. 
    
    Define the polynomial:
    \begin{align*}
        \tilde{Q}^{\,d} \;=\; \left(\sum_{(i,j) \in \calM} x_i \cdot x_j\right)^d = \left(\sum_{i \in [t]} y_i \cdot z_i\right)^d.
    \end{align*}
     In particular, ${\tilde Q}^{\,d}$ is obtained from $Q_G^d$ by setting to zero the variables which are not in the matching $\calM$. Hence, the rank of the Nisan matrix corresponding to $\tilde{Q}$ is a lower bound on the evaluation dimension of $Q$ w.r.t. the partition $(Y,Z)$. 
    
    By construction, for every monomial $m_Y$ of degree exactly $d$ over $Y$, there exists a unique monomial $m_Z$ over $Z$ such that the coefficient of $m_Y \cdot m_Z$ in $\tilde{Q}^{\,d}$ is nonzero (cf.~\autoref{ob:rochar-example}).
    Therefore,
    $$
        \rank\left(M_{Y,Z}(Q_G^d)\right) \;\geq\; \rank\left(M_{Y,Z}\left(\tilde{Q}^{\,d}\right)\right) \;\geq\; \binom{d + t-1}{t-1}.
    $$
    Applying \autoref{thm:roabp-char}, we obtain the desired lower bound on the width of any $\RO$ computing $Q_G^d$.
\end{proof}

\subsection{Hardness of computing resultant and discriminant}

We design a polynomial that is simple for $\RO$, but which turns out to be difficult for $\RO$ when one computes its discriminant.
This, in turn, immediately implies that computing resultant is also hard for $\RO$.

\discriminant*
\begin{proof}
    Let $g$ be an $(n-1)$-variate polynomial to which the lower bound in \autoref{cor:power-roabp} is applicable. 
    Fix any $d = \Omega(n)$. Define $$f \coloneqq y^d - g(\vecx) \cdot y.$$
    Then we have $\partial_y f = d\cdot y^{d-1} - g$.
    It is easy to see that the roots of $f$ are $\alpha_0 = 0$ and $\alpha_i = \omega^i \cdot g^{1/(d-1)}$ for $1 \leq i \leq d-1$, where $\omega$ is a primitive $(d-1)$-th root of unity. Here we work over a suitable field extension of the base field $\F$ to ensure that we have an $(d-1)$-th root of unity.

    The discriminant of $f$ is defined as the resultant of $f$ and $\partial_y f$ with respect to $y$, i.e., $\operatorname{Disc}_y(f) = \operatorname{Res}_y(f, \partial_y f)$. Then using \autoref{ob:res-product} we can compute:
    \begin{align*}
        \operatorname{Disc}_y(f) \;&=\;\prod_{i = 0}^{d-1} \partial_y f(\alpha_i)
        \;=\; -g \cdot \prod_{i = 1}^{d-1} \left( d -1 \right) \cdot g\\
        \;&=\; -(d-1)^{d-1} g^d.
    \end{align*}

    Thus, computing the discriminant of $f$ amounts to powering the polynomial $g$, for which we have the required lower bound by \autoref{cor:power-roabp}. The upper bound on the $\RO$ complexity of $f$ follows from the one for $g.$
\end{proof}

\section{Conclusions and Open Problems}

In this work, we proved that a width-$w$ $\RO$ computes a polynomial whose irreducible factor requires $\RO$s of width at least $w^{\log d}$, yielding a quasipolynomial separation. This showed that $\RO$s are not closed under factoring (see \autoref{sec:factor-non-closure}). A natural next step is to search for polynomials that exhibit an exponential separation between the $\RO$ complexity of a polynomial and that of its factor. 

Our non-factor closure proof relied on the idea that polynomials that are hard for the simpler sparse model but easy for $\RO$s can be transformed, using simple gadgets, into polynomials that are hard even for $\RO$s. This raises an intriguing question about the scope of such hardness lifting. Specifically, given a polynomial $f(\vecx)$ of sparsity $s$, can we always lift using a gadget $\phi$ such that the composed polynomial $f(\phi \circ \vecx)$ requires an $\RO$ of size $\Omega(s)$?

One consequence of our study of graph-based polynomials and symmetric compositions is the proof that powering is hard for $\RO$s (see \autoref{sec:hard-power}). This naturally raises the question in the other direction: does there exist a polynomial $f \coloneqq g^d$ that is easy to compute by an $\RO$, while $g$ is hard for $\RO$? 
An affirmative answer would, once again, stand in sharp contrast to other models such as circuits, algebraic branching programs, and formulas, where low complexity of $f$ leads to low complexity of $g$.
Interestingly, the analogous question for sparse polynomials was answered affirmatively for $d = 2$ in classical works by R{\'{e}}nyi~\cite{Renyi1951} and Erd{\H{o}}s~\cite{Erdos1949}, and was subsequently extended to arbitrary $d$ in later works~\cite{Ver1949, CD1994}.

\paragraph*{Acknowledgments.} \textit{The research in this paper was carried out during a visit of some of the authors to Université Savoie Mont Blanc. The authors thank the Laboratoire de Mathématiques (LAMA) at Université Savoie Mont Blanc, which is supported by the AAP GAFA project, for their hospitality and the conducive environment for research.}

\printbibliography

\end{document}